\documentclass[12pt]{amsart}

\usepackage{fullpage}
\usepackage{epsfig}
\usepackage{color}
\usepackage{extarrows}
\usepackage{amsfonts}
\usepackage{amssymb}
\usepackage{amsmath}
\usepackage{graphicx}

\newtheorem{thm}{Theorem}[section]
\newtheorem{lem}{Lemma}[section]

\newtheorem{cor}{Corollary}[section]

\newtheorem{defi}{Definition}[section]
\numberwithin{equation}{section}

\def\f#1{{\mathbb{F}}_{#1}}

\def\gbinom#1#2{\begin{bmatrix}#1\\#2\end{bmatrix}_q}

\begin{document}

\title[RGHWs of Cyclic Codes]{Relative Generalized Hamming Weights of Cyclic Codes}
\author[J. Zhang]{Jun Zhang}
\address{School of Mathematical Sciences, Capital Normal University, Beijing 100048, P.R. China}
\email{junz@cnu.edu.cn}
\author[K. Feng]{Keqin Feng}
\address{Department of Mathematics, Tsinghua University, Beijing 100084, P.R. China}
\email{kfeng@math.tsinghua.edu.cn}

\thanks{The second author is supported by NSFC No.11471178 and the Tsinghua National Lab. for Information Science and Technology.}

\begin{abstract}
Relative generalized Hamming weights (RGHWs) of a linear code respect to a linear subcode determine the security of the linear ramp secret sharing scheme based on the code.
They can be used to express the information leakage of the secret when some keepers of shares are corrupted. Cyclic codes are an interesting type of linear codes and have wide applications in communication and storage systems. In this paper, we investigate the RGHWs of cyclic codes with two nonzeros respect to any of its irreducible cyclic subcodes. Applying the method in~\cite{YLFL}, we give two formulae for RGHWs of the cyclic codes. As applications of the formulae, explicit examples are computed.
\end{abstract}
\maketitle


\section{Introduction}
A secret sharing scheme splits a secret into several pieces which are distributed to participants, so that only specified subsets of participants can recover
the secret. Shamir~\cite{Sh79} proposed the first secret sharing scheme in which the subsets of participants unable to reconstruct the secret can not get any information about the secret. A secret sharing scheme with this property is called a perfect scheme. Later, realizing that Shamir's construction is indeed from Reed-Solomon codes~\cite{MS81}, Shamir's construction was generalized to secret sharing schemes based on linear codes~\cite{CC06,Ma93,CLX}. All these constructions are perfect schemes. The security of secret sharing schemes based on linear codes are completely characterized by the minimal codewords in the dual codes~\cite{Ma93}. Non-perfect secret sharing schemes were proposed~\cite{BC85} where the subsets of participants unable to reconstruct the secret could get some information about the secret. The term ``tamp" is used in this scenario. It was shown in~\cite{GMMRL} that in secret sharing schemes based on linear codes, the amount of information leakage of the secret when some participants are corrupted is completely determined by the RGHWs of the involved codes.

The concept of RGHWs of linear codes is an extension of generalized Hamming weights (GHW) of linear codes, but the studies of RGHWs can not be substituted by those of GHWs. The concept of GHWs of linear codes was first introduced by Wei~\cite{We91} to study the linear codes over wire-tap channel of type II and was later used by Ozarow and Wyner~\cite{OW85} in cryptography to characterize the performance of linear codes when used over such a channel. The GHWs of linear codes have been used in many other applications~\cite{Gur03,Hel95,JL07,KTFS93,NYZ11}, so the study of GHWs of linear codes has attracted much attention in the past two decades. Comparing with GHWs, the RGHWs are defined on a pair of linear codes (a linear codes and one of its subcodes); the RGHWs have only one more restriction which makes they are often larger than the corresponding GHWs; and they should have applications parallelly to those of GHWs which need to be found in the furture.

The concept of RGHWs was proposed by Luo et al.~\cite{LMVK05} in their study of communication over the wire-tap channel of type II. They proved the equivalence between RGHW and the corresponding relative dimension/length profile (RDLP), similarly as the relation between the dimension/length profile and the GHW demonstrated by Forney~\cite{For94}. Bounds on RGHWs and constructions of good codes are studied in~\cite{ZLVD11,ZLD13,Geil15}.


So far, only few classes of linear codes have been examined for their RGHW/RDLP. In~\cite{LCL08}, RGHWs of almost all 4-dimensional linear code respect to its subcodes are determined by using techniques in finite projective geometry. In~\cite{GMMRL}, the authors estimated RGHWs of general linear codes by Feng-Rao approach~\cite{FR93} and in particular they checked their bounds for one-point Hermitian codes. In~\cite{MG14}, RGHWs of $q$-ary Reed-Muller codes are studied, closed formula expressions for $q$-ary Reed-Muller codes in two variables were presented, and simple and low complexity algorithm to determine RGHWs of general $q$-ary Reed-Muller codes were given by extension of Heijnen and Pellikaan method~\cite{HP98}.

In this paper, we consider a special class of cyclic codes and give two formulae for their RGHWs.
Section~\ref{sec2} reviews the definition of RGHW and introduces the cyclic codes we investigate. The one-to-one correspondence between
the subspaces of the cyclic code and the subspaces of the direct product of two big fields is established. Equipped with this correspondence,
one inner product on the direct product space of two big fields is defined to transfer the RGHW problem to a counting problem on the dual space. Hence, we obtain
our first formula of RGHWs. By employing exponential sums, the common method in computing the weight distributions of cyclic codes, we get the second
formula of RGHWs. As applications of the two formulae, explicit examples are calculated, and explicit formulae are given.

\section{Relative Generalized Hamming Weights of Cyclic Codes}\label{sec2}
In this section, we give two formulae for RGHWs of cyclic codes of two nonzeros respect to one of its irreducible
cyclic subcodes.

We first introduce some notations valid for the whole paper.
\begin{itemize}
  \item Let $\f{q}$ be the finite field of $q$ elements with characteristic
$p$. Denote an extension field of $\f{q}$ by $\f{Q}$, and the trace map from $\f{Q}$
to $\f{q}$ by ${T_q^Q}$. Precisely, for any $a\in\f{Q}=\f{q^m}$, the trace ${T_q^Q}$ of $a$ is ${T_q^Q}(a)=a+a^q+\cdots+a^{q^{m-1}}$.
  \item Let $V\subset \f{q}^N$ be a $k$-dimensional linear space over $\f{q}$.
  We also call $V$ an $[N,k]$ linear code over $\f{q}$. For any $1\leq j \leq k$, denote
  by $\gbinom{V}{j}$ the set of all $j$-dimensional linear subspaces of $V$.
  \item For any $D\in \gbinom{V}{j}$, define the support of $D$ to be the set of locations at which
  all the vectors in $D$ are zeros. Denote by $\mathrm{Supp}(D)$ the support of $D$. That is,
   \[
      \mathrm{Supp}(D)=\{i\,:\, 1\leq i\leq N,\, v_i\neq 0\,\textrm{ for some $v=(v_1,\cdots,v_N)\in D$}\}.
  \]
\end{itemize}

\begin{defi}
Let $C_1\subset C_2\subset \f{q}^N$ be two $\f{q}$ linear codes of dimension $k_1,k_2$, respectively. For
$1\leq j \leq k_2-k_1$, the $j$th RGHW of $C_2$
respect to its subcode $C_1$ is defined to be the minimum size of $\mathrm{Supp}(D)$ where $D$ runs through all $j$-dimensional linear subspaces
of $C_2$ among which everyone has intersection $\{0\}$ with $C_1$. Denote it by $M_j(C_2,C_1)$. That is,
\[
   M_j(C_2,C_1)=\min\left\{|\mathrm{Supp}(D)|\,:\, D\in \gbinom{C_2}{j},\, D\cap C_1=\{0\}\right\}.
\]
\end{defi}

Note that if the restriction $D\cap C_1=\{0\}$ is released in the definition of RGHWs, then it becomes the corresponding GHW $d_j$ of the code $C_2$.

Now we introduce the cyclic codes considered in this paper. Let $\f{Q}$ be a finite field with $Q=q^m$. Suppose $\alpha_1, \alpha_2$ are two elements in $\f{Q}^*$ which are not conjugate to each
other over $\f{q}$. Let $n_1, n_2$ be the orders of $\alpha_1$ and $\alpha_2$, respectively. Let $d=\gcd(n_1,n_2)$ and $n=\frac{n_1n_2}{d}$.
Suppose $\gamma_1$ and $\gamma_2$ are primitive elements of the fields $\f{q}(\alpha_1)=\f{Q_1}=\f{q^{k_1}}$ and $\f{q}(\alpha_2)=\f{Q_2}=\f{q^{k_2}}$,
respectively, then $\alpha_i=\gamma_i^{e_i}$ and $Q_i-1=e_in_i$ for $i=1,2$.

With the above setting, we have two irreducible cyclic codes
\[
   C_i=\{c(\beta_i)=(T_q^{Q_i}(\beta_i),T_q^{Q_i}(\beta_i\alpha_i),\cdots,T_q^{Q_i}(\beta_i\alpha^{n_i-1}))\,:\,\beta_i\in\f{Q_i}\}
\]
for $i=1,2$. It is easy to see that $C_i$ has parameters $[n_i,k_i]$, and by Delsarte's theorem~\cite{Del75}, the parity check polynomial $h_i(x)\in \f{q}[x]$ of $C_i$ is the minimal (also irreducible) polynomial of
$\alpha_i^{-1}$ over $\f{q}$, for $i=1,2$.
The Hamming weight distributions of irreducible cyclic codes and general cyclic codes
are studied extensively in the literature. Recently, the GHWs of irreducible cyclic codes are studied by the second author of this paper~\cite{YLFL} and the GHWs of more special cyclic codes are presented in~\cite{XLG15}.
Using the method in~\cite{YLFL}, we compute the RGHWs of the following two cyclic codes:
\begin{align}\label{code1}
     C=\{c(\beta_1,\beta_2)\,:\,\beta_1\in\f{Q_1}\,,\beta_2\in\f{Q_2}\},
\end{align}
where
\begin{align*}
     c(\beta_1,\beta_2)=(T_q^{Q_1}(\beta_1)+T_q^{Q_2}(\beta_2),T_q^{Q_1}(\beta_1\alpha_1)+T_q^{Q_2}(\beta_2\alpha_2),\cdots,\\
  T_q^{Q_1}(\beta_1\alpha^{n-1}+T_q^{Q_2}(\beta_2\alpha^{n-1})),
\end{align*}
and its subcode
\begin{align}\label{code2}
  C'=\{c(0,\beta_2)\,:\,\beta_2\in\f{Q_2}\}.
\end{align}
Note that the codes $C$ and $C'$ have parameters $[n,k_1+k_2]$ and $[n,k_2]$, respectively, and the codeword $c(0,\beta_2)\in C'$ is the repeation of the codeword $c(\beta_2)\in C_2$ by $n\over n_2$ times. Again by Delsarte' theorem, the cyclic code $C$ has parity check polynomial $h_1(x)h_2(x)$. Such a cyclic code $C$ is often called cyclic code with two nonzeros, since the parity check polynomial factors through two distinct irreducible polynomials.

\subsection{The First Formula}
We first characterize the condition that subspaces of $C$ have zero-intersection with $C'$, then give our first formula for $M_j(C,C')$.

From the definition of the cyclic code $C$, we have a canonical isomorphism of $\f{q}$-linear spaces:
\[
   \f{Q_1}\times\f{Q_2}\rightarrow C, \qquad (\beta_1, \beta_2)\mapsto c(\beta_1, \beta_2).
\]
So we get a one-one correspondence
\begin{equation}\label{11map}
    \varphi: \gbinom{C}{j}\rightarrow \gbinom{\f{Q_1}\times\f{Q_2}}{j}.
\end{equation}

For any $D\in\gbinom{C}{j}$, under this correspondence, we have
\[
   D\cap C'=\{0\}\quad\Leftrightarrow\quad \varphi(D)\cap (\{0\}\times\f{Q_2})=\{0\}.
\]
If we denote by $\pi_i$ the projection of $H=\varphi(D)$ at the $i$th coordinate:
\begin{align*}
    \pi_1: H\rightarrow \f{Q_1}, \qquad (\beta_1, \beta_2)\mapsto \beta_1;\\
    \pi_2: H\rightarrow \f{Q_2}, \qquad (\beta_1, \beta_2)\mapsto \beta_2,
\end{align*}
then
\[
  H\cap (\{0\}\times\f{Q_2})=\{0\}\quad\Leftrightarrow\quad \ker(\pi_1)=\{0\}\quad\Leftrightarrow\quad
  \mathrm{Im}(\pi_1)\in\gbinom{\f{Q_1}}{j}.
\]
Furthermore, an inner product on $\f{Q_1}\times\f{Q_2}$ is introduced as follows. For any $(x_1,y_1),(x_2,y_2)\in \f{Q_1}\times\f{Q_2}$,
the inner product between the two elements is defined to be
\[
   <(x_1,y_1),(x_2,y_2)>=T_q^{Q_1}(x_1x_2)+T_q^{Q_2}(y_1y_2)\in \f{q}.
\]
With respect to this inner product, the dual space of $H$ is defined to be
\[
    H^{\bot}=\{v\in\f{Q_1}\times\f{Q_2}\,:\,<v,w>=0,\,\forall w\in H\}\in \gbinom{\f{Q_1}\times\f{Q_2}}{k_1+k_2-j}.
\]
So the intersection property can be interpreted as following
\[
  H\cap (\{0\}\times\f{Q_2})=\{0\}\Leftrightarrow H^{\bot}+(\f{Q_1}\times\{0\})=\f{Q_1}\times\f{Q_2}\Leftrightarrow\pi_2(H^{\bot})=\f{Q_2}.
\]
Now, with the preparation above, the first formula is presented in the following theorem.
\begin{thm}\label{thm1}
Let $C$ and $C'$ be the cyclic codes given by~(\ref{code1})
and~(\ref{code2}), respectively. For $1\leq j\leq k_1$, we have
\[
   M_j(C,C')=n-N_j,
\]
where
\[
N_j=\max\left\{|H\cap <(\alpha_1,\alpha_2)>|\,:\, H\in\gbinom{\f{Q_1}\times\f{Q_2}}{k_1+k_2-j},\, \pi_2(H)=\f{Q_2}\right\},
\]
and $<(\alpha_1,\alpha_2)>$ denotes the subgroup $\{(\alpha_1^i,\alpha_2^i)\,:\, 1\leq i\leq n\}$ generated by $(\alpha_1,\alpha_2)$.
\end{thm}
\begin{proof}
By definition, we have
\begin{align*}
     M_j(C,C')&=\min\left\{|\mathrm{Supp}(D)|\,:\, D\in \gbinom{C}{j},\, D\cap C'=\{0\}\right\}\\
                 &=n-\max\left\{N_j(D)\,:\, D\in \gbinom{C}{j},\, D\cap C'=\{0\}\right\},
\end{align*}
where
\begin{align*}
    N_j(D)&=|\{i\,:\, 0\leq i\leq n-1,\, T_q^{Q_1}(\beta_1\alpha_1^i)+T_q^{Q_2}(\beta_2\alpha_2^i)=0,\, \forall (\beta_1,\beta_2)\in D \}|\\
          &=|\{i\,:\, 0\leq i\leq n-1,\, <(\beta_1,\beta_2),(\alpha_1^i,\alpha_2^i)>=0,\, \forall (\beta_1,\beta_2)\in \varphi(D) \}|\\
          &=|\{i\,:\, 0\leq i\leq n-1,\, (\alpha_1^i,\alpha_2^i)\in \varphi(D)^{\bot} \}|.
\end{align*}
Together with the discussion previous of the theorem, one can easily obtain the theorem.

\end{proof}

\subsection{Applications of the First Formula}

As one application, taking $q=2, Q_1=2^{k_1}, Q_2=2^{k_2},
e_1=e_2=1$ such that $\gcd(k_1, k_2)=1, k_1, k_2\geq 2$,
$\alpha_i=\gamma_i$ is the primitive element of $\f{Q_i}$ for
$i=1,2$. In this case, $C_1, C_2$ are cyclic codes from
$m$-sequences of degree $k_1, k_2$, respectively. Since $\gcd(k_1,
k_2)=1$, it follows that $\gcd(n_1,n_2)=\gcd(Q_1-1,Q_2-1)=1$. So the length
$n=n_1n_2$ and the group generated by $(\alpha_1,\alpha_2)$
\[
   <(\alpha_1,\alpha_2)>=<\alpha_1>\times<\alpha_2>=\f{Q_1}^*\times\f{Q_2}^*.
\]
Hence, for any $H\in\gbinom{\f{Q_1}\times\f{Q_2}}{k_1+k_2-j}\,(1\leq j\leq k_1)$, we have
\begin{align*}
   &|H\cap <(\alpha_1,\alpha_2)>|\\
   =&|H\cap (\f{Q_1}^*\times\f{Q_2}^*)|\\
   =&|\{(\beta_1,\beta_2)\in H\,:\, \beta_1\beta_2\neq 0\}|\\
   =&q^{k_1+k_2-j}-|\pi_1^{-1}(0)|-|\pi_2^{-1}(0)|+1.
\end{align*}
Taking account of the property $\pi_2(H)=\f{Q_2}$, $\pi_2^{-1}(0)$ is an $\f{q}$-linear subspace of $H$ which has dimension
\[
  (k_1+k_2-j)-k_2=k_1-j
\]
from the exact sequence of $\f{q}$-linear spaces
\[
   \{0\}\rightarrow \pi_2^{-1}(0)\rightarrow H\rightarrow \pi_2(H)\rightarrow \{0\}.
\]
So
\[
  |H\cap <(\alpha_1,\alpha_2)>|=q^{k_1+k_2-j}-q^{k_1-j}-|\pi_1^{-1}(0)|+1.
\]
By the same reason, $\dim_{\f{q}}(\pi_1^{-1}(0))=k_1+k_2-j-
   \dim_{\f{q}}(\pi_1(H))\geq k_2-j$.
On the other hand,
$
 \dim_{\f{q}}(\pi_1^{-1}(0))\leq
 \dim_{\f{q}}(\{0\}\times\f{Q_2})=k_2$. So,
 \[
  k_2-j\leq \dim_{\f{q}}(\pi_1^{-1}(0))\leq k_2,
 \]
 and hence
 \[
 q^{k_1+k_2-j}-q^{k_1-j}-q^{k_2}+1\leq |H\cap <(\alpha_1,\alpha_2)>|\leq
 q^{k_1+k_2-j}-q^{k_1-j}-q^{k_2-j}+1.
 \]

 Let $H_1$ be the linear space spanning of $v_i=(\alpha_1^i,\alpha_2^i), 0\leq i\leq
 k_1-1$, over $\f{q}$. Since $\alpha_1^i, 0\leq i\leq
 k_1-1$, are linearly independent over $\f{q}$, we have
 $$\dim_{\f{q}}(H_1)=k-1.$$
 It is easy to see that the projection $\pi_1: H_1\rightarrow
 \f{Q_1}$ has kernel $\ker(\pi_1)=\{0\}$.

When $k_2\geq k_1$, the projection $\pi_2(H_1)$ is
$k_1$-dimensional $\f{q}$-linear subspace of $\f{Q_2}$ generated
by $\alpha_2^i, 0\leq i\leq k_1-1$. For $1\leq j\leq k_1$,
$m=k_1+k_2-j\geq k_2$, so there is some $(m-k_1)$-dimensional
$\f{q}$-linear subspace $V$ of $\f{Q_2}$ such that
\[
   V+\pi_2(H_1)=\f{Q_2}.
\]
Consider the subspace $H=H_1+(0,V)$ of $\f{Q_1}\times\f{Q_2}$. It
follows from $\ker(\pi_1)=\{0\}$ that
$$H_1\cap(0,V)={0},$$
so the sum in $H$ is a direct sum. Then
\[
  \dim_{\f{q}}(H)=k_1+(m-k_1)=m.
\]
That is, $H\in\gbinom{\f{Q_1}\times\f{Q_2}}{k_1+k_2-j}$,
$\pi_2(H)=V+\pi_2(H_1)=\f{Q_2}$, and $\dim_{\f{q}}(\pi_1^{-1}(0))$
achieves the minimum $m-k_1=k_2-j$. So
\[
  N_j=q^{k_1+k_2-j}-q^{k_1-j}-q^{k_2-j}+1.
\]

When $1\leq j\leq k_2<k_1$, the same construction as above can
show that $N_j$ has the same formula
\[
  N_j=q^{k_1+k_2-j}-q^{k_1-j}-q^{k_2-j}+1.
\]

When $k_2<j\leq k_1$, we have $k_2\leq m=k_1+k_2-j< k_1$.
Taking $H$ to be the $\f{q}$-linear space spanning by
$(\alpha_1^i,\alpha_2^i), 0\leq i\leq m-1$, as $m\leq k_1$, we
know $\dim_{\f{q}}(H)=m$. From $m\geq k_2$, it follows that
$\pi_2(H)=\f{Q_2}$. In this case, $\pi_1^{-1}(0)=\{0\}$ achieves
the minimal dimension. So
\[
  N_j=q^{k_1+k_2-j}-q^{k_1-j}.
\]

In conclusion, we have the following corollary:
\begin{cor}
 Let $\alpha_1,\alpha_2$ be the primitive elements of $\f{Q_1}$ and
 $\f{Q_2}$, respectively, where
$Q_1=2^{k_1}, Q_2=2^{k_2}, \gcd(k_1, k_2)=1$. Let $C, C'$ be the
cyclic codes defined by~(\ref{code1}) and~(\ref{code2}),
respectively. For $1\leq j\leq k_1$, we have
\[
   M_j(C,C')=(Q_1-1)(Q_2-1)-N_j,
\]
where if $k_1\leq k_2$, then
\[
N_j=q^{k_1+k_2-j}-q^{k_1-j}-q^{k_2-j}+1;
\]
if $k_2< k_1$, then
\begin{equation*}
   N_j=\left\{%
\begin{array}{ll}
    q^{k_1+k_2-j}-q^{k_1-j}-q^{k_2-j}+1, & \hbox{if $1\leq j\leq k_2$;} \\
   q^{k_1+k_2-j}-q^{k_1-j}, & \hbox{if $k_2<j\leq k_1$.} \\
\end{array}%
\right.
\end{equation*}
\end{cor}

\subsection{The Second Formula}
The study of exponential sums, no matter the estimations or the
exact values, is one important topic in number theory. Exponential
sums, as a tool, are wildly used in coding theory, such as
the proof of MacWilliams identity, computing the minimum distance
and weight distributions of codes, especially those of cyclic
codes, etc. In this subsection, the exponential sums are used to
present a formula for the RGHWs of
the cyclic codes we consider in this paper.

We first review the Gauss sums over finite fields. A character of the abelian
group $(G,+)$ is a group homomorphism from $G$ to the multiplicative group of nonzero complex numbers $\mathbb{C}^*$.
Additive/multiplicative characters of a finite field $\f{q}$ are characters of the field considered
as the additive group $(\f{q},+)$ or multiplicative group $(\f{q}^*,\cdot)$, respectively. For
any multiplicative character $\chi$ of $\f{q}$ and any $\beta\in\f{q}$, the Gauss sum of $\chi$ and $\beta$
is defined to be the sum
\[
    G_q(\chi;\beta)=\sum_{x\in\f{q}^*}\chi(x)\zeta_{p}^{T^q_p(\beta x)},
\]
where $\zeta_l=e^{2\pi\sqrt{-1}/l}$ is the primitive $l$th root of unity, for any positive integer $l$.
We always denote
\[
   G_q(\chi)=G_q(\chi;1).
\]
In particular, we have the following relationship
\begin{equation}\label{relation}
     G_q(\chi;\beta)=\bar{\chi}(\beta)G_q(\chi)
\end{equation}
where $\bar{\chi}$ is the conjugation of $\chi$ defined by $\bar{\chi}(x)=\overline{\chi(x)}$ for any $x\in\f{q}$.
Also note that
$$ G_q(\chi;0)=\sum_{x\in\f{q}^*}\chi(x)=\left\{
                                           \begin{array}{ll}
                                             q-1, & \hbox{$\chi=1$;} \\
                                             0, & \hbox{otherwise.}
                                           \end{array}
                                         \right.
$$
is the complete sum of the character $\chi$. In general, we have the following character sum over finite abelian groups which
can be viewed as the duality property.
\begin{lem}\label{lem1}
Let $\theta$ be a primitive element of $\f{q}$. Let $\chi$ be the multiplicative character of $\f{q}$ defined by $\chi(\theta)=\zeta_e$.
For $\alpha=\theta^e$, and any $x\in\f{q}^*$, we have
\[
    \sum_{\lambda=0}^{e-1}\chi^{\lambda}(x)=\left\{
                                              \begin{array}{ll}
                                                e, & \hbox{if $x\in <\alpha>$;} \\
                                                0, & \hbox{otherwise,}
                                              \end{array}
                                            \right.
\]
where $<\alpha>=\{1,\alpha,\alpha^2,\cdots,\alpha^{e-1}\}$ is the subgroup generated by $\alpha$.
\end{lem}

From the proof of Theorem~\ref{thm1}, in order to compute
$M_j(C,C')$, it is enough to compute
\[
    N_j(D)=|\{i\,:\, 0\leq i\leq n-1,\, T_q^{Q_1}(\beta_1\alpha_1^i)+T_q^{Q_2}(\beta_2\alpha_2^i)=0,\, \forall (\beta_1,\beta_2)\in D \}|\\
\]
then take the maximum of $N_j(D)$ where $D$ runs over
$\gbinom{C}{j}$ with $D\cap C'=\{0\}$.

For $1\leq j\leq k_1$, any $D\in \gbinom{C}{j}$ with $D\cap
C'=\{0\}$, take a basis
$\{(\beta_1^{(\lambda)},\beta_2^{(\lambda)})\,:\,1\leq \lambda\leq
j\}$ for $H=\varphi(D)$ over $\f{q}$. Then
\begin{align*}
    N_j(D)=&|\{i\,:\, 0\leq i\leq n-1,\, T_q^{Q_1}(\beta_1\alpha_1^i)+T_q^{Q_2}(\beta_2\alpha_2^i)=0,\, \forall (\beta_1,\beta_2)\in H \}|\\
          =&|\{i\,:\, 0\leq i\leq n-1,\, T_q^{Q_1}(\beta_1^{(\lambda)}\alpha_1^i)+T_q^{Q_2}(\beta_2^{(\lambda)}\alpha_2^i)=0,\, \forall 1\leq \lambda\leq j \}|\\
          =&\frac{1}{q^j}\sum_{i=0}^{n-1}\prod_{\lambda=1}^{j}\sum_{x_t\in\f{q}}\zeta_p^{T_p^{q}(x_t(T_q^{Q_1}(\beta_1^{(\lambda)}\alpha_1^i)+T_q^{Q_2}(\beta_2^{(\lambda)}\alpha_2^i)))}\\
          =&\frac{1}{q^j}\sum_{i=0}^{n-1}\sum_{x_1,\cdots,x_j\in\f{q}}\zeta_p^{T_p^{Q_1}(\alpha_1^i(x_1\beta_1^{(1)}+\cdots+x_j\beta_1^{(j)}))+T_p^{Q_2}(\alpha_2^i(x_1\beta_2^{(1)}+\cdots+x_j\beta_2^{(j)}))}\\
          =&\frac{1}{q^j}\sum_{i=0}^{n-1}\sum_{(\beta_1,\beta_2)\in H}\zeta_p^{T_p^{Q_1}(\beta_1\alpha_1^i)+T_p^{Q_2}(\beta_2\alpha_2^i)}\\
          =&\frac{n}{q^j}+\frac{1}{q^j}(A+B),
\end{align*}
where
\begin{align*}
         A&=\sum_{i=0}^{n-1}\sum_{(\beta_1,0)\in H,\beta_1\neq 0}\zeta_p^{T_p^{Q_1}(\beta_1\alpha_1^i)}\\
          &=\frac{n}{n_1}\sum_{(\beta_1,0)\in H,\beta_1\neq 0}\sum_{x\in <\alpha_1>}\zeta_p^{T_p^{Q_1}(\beta_1x)}\\
          &\xlongequal{(1)}\frac{n}{e_1n_1}\sum_{(\beta_1,0)\in H,\beta_1\neq 0}\sum_{x\in \f{Q_1}^*}\zeta_p^{T_p^{Q_1}(\beta_1x)}\sum_{\lambda=0}^{e_1-1}\chi^{\lambda}(x)\\
          &\xlongequal{(2)}\frac{n}{e_1n_1}\sum_{(\beta_1,0)\in H,\beta_1\neq 0}\sum_{\lambda=0}^{e_1-1}\bar{\chi}^{\lambda}(\beta_1)G_{Q_1}(\chi^{\lambda})\\
          &\xlongequal{(3)}\frac{n}{e_1n_1}\sum_{\lambda=0,\chi^{\lambda}(\f{q}^*)=1}^{e_1-1}G_{Q_1}(\chi^{\lambda})\sum_{(\beta_1,0)\in H,\beta_1\neq 0}\bar{\chi}^{\lambda}(\beta_1),
\end{align*}
where $\chi$ is the multiplicative character of $\f{Q_1}$ such that $\chi(\gamma)=\zeta_{e_1}$;
and
\[
B=\sum_{(\beta_1,\beta_2)\in H,\beta_1\beta_2\neq 0}\sum_{i=0}^{n-1}\zeta_p^{T_p^{Q_1}(\beta_1\alpha_1^i)+T_p^{Q_2}(\beta_2\alpha_2^i)}.
\]
The equality (1) follows from the relationship~(\ref{relation}). The equality (2) follows from Lemma~\ref{lem1}. The equality (3) follows from
the property~\footnote{The proof is very easy. Take any $\theta\in\f{q}$ such that $\psi(\theta)\neq 1$, then
\[
  \psi(\theta)\sum_{x\in H}\psi(x)=\sum_{x\in H}\psi(\theta x)=\sum_{x\in\theta H}\psi(x)=\sum_{x\in H}\psi(x)
\]
which implies $\sum_{x\in H}\psi(x)=0$.
}: if $\psi$ is a multiplicative character of $\f{Q}$ which is non-trivial over the subfield $\f{q}$, and $H$ is an $\f{q}$-linear subspace of $\f{Q}$, then the imcomplete sum $\sum_{x\in H}\psi(x)=0$.

So by using the above expression of $N_j(D)$, we give another formula for the RGHW $M_j(C,C')$.

 \begin{flushleft}
 \textbf{Discussion of $A$ and $B$.}
\end{flushleft}

We now return to the property $\chi^{\lambda}(\f{q}^*)=1\,(0\leq \lambda\leq e_1-1)$. Since $\gamma_1^{\frac{Q_1-1}{q-1}}$ is a primitive element of $\f{q}$, the property $\chi^{\lambda}(\f{q}^*)=1$ is equivalent to
 $$1=\chi^{\lambda}(\gamma_1^{\frac{Q_1-1}{q-1}})=\zeta_{e_1}^{\lambda\frac{Q_1-1}{q-1}},$$
 which is also equivalent to
\[
  e_1\mid\lambda\frac{Q_1-1}{q-1}, \textrm{i.e.,}\,\frac{e_1}{e'_1}\mid\lambda,\,\textrm{where}\, e'_1=\gcd(e_1,\frac{Q_1-1}{q-1}).
\]
So $\lambda$ has the form $\lambda=\frac{e_1}{e'_1}\tau,\,0\leq \tau\leq e'_1-1$. Denote $\psi=\chi^{{e_1}/{e'_1}}$, then $\psi(\gamma_1)=\zeta_{e'_1}$ and
\begin{align*}
         A=&\frac{n}{e_1n_1}\sum_{\lambda=0,\chi^{\lambda}(\f{q}^*)=1}^{e_1-1}G_{Q_1}(\chi^{\lambda})\sum_{(\beta_1,0)\in H,\beta_1\neq 0}\bar{\chi}^{\lambda}(\beta_1)\\
         =&\frac{n}{e_1n_1}\sum_{\tau=0}^{e'_1-1}G_{Q_1}(\psi^{\tau})\sum_{(\beta_1,0)\in H,\beta_1\neq 0}\bar{\psi}^{\tau}(\beta_1)\\
         =&-\frac{n|(\f{Q_1}^*,0)\cap H|}{e_1n_1}+\frac{n}{e_1n_1}\sum_{\tau=1}^{e'_1-1}G_{Q_1}(\psi^{\tau})\sum_{(\beta_1,0)\in H,\beta_1\neq 0}\bar{\psi}^{\tau}(\beta_1).
\end{align*}
In particular, if $e'_1=1$, then
\[
  A=-\frac{n|(\f{Q_1}^*,0)\cap H|}{Q_1-1}.
\]

We now study the exponential sum $B$ provided that $\gcd(n_1,n_2)=1$. In this case, the length $n$ of the code $C$ equals $n_1n_2$. For any $0\leq i\leq n-1$, there is a unique pair $(i_1,i_2)$ such that $0\leq i_1\leq n_1-1$, $0\leq i_2\leq n_2-1$ and
\[
  i=i_1n_1+i_2n_2\qquad\mod n.
\]
We can compute
\begin{align*}
         B=&\sum_{(\beta_1,\beta_2)\in H,\beta_1\beta_2\neq 0}\sum_{i=0}^{n-1}\zeta_p^{T_p^{Q_1}(\beta_1\alpha_1^i)+T_p^{Q_2}(\beta_2\alpha_2^i)}\\
         =&\sum_{(\beta_1,\beta_2)\in H,\beta_1\beta_2\neq 0}\left(\sum_{i_1=0}^{n_1-1}\zeta_p^{T_p^{Q_1}(\beta_1\alpha_1^{n_2i_1})}\right)\left(\sum_{i_2=0}^{n_2-1}\zeta_p^{T_p^{Q_2}(\beta_2\alpha_2^{n_1i_2})}\right)\\
          =&\sum_{(\beta_1,\beta_2)\in H,\beta_1\beta_2\neq 0}\left(\sum_{x_1\in<\alpha_1>}\zeta_p^{T_p^{Q_1}(\beta_1x_1)}\right)\left(\sum_{x_2\in<\alpha_2>}\zeta_p^{T_p^{Q_2}(\beta_2x_2)}\right)\\
          =&\frac{1}{e_1e_2}\sum_{(\beta_1,\beta_2)\in H,\beta_1\beta_2\neq 0}\sum_{\lambda_1=0}^{e_1-1}\sum_{\lambda_2=0}^{e_2-1}G_{Q_1}(\chi_1^{\lambda_1})G_{Q_2}(\chi_2^{\lambda_2})\bar{\chi_1}^{\lambda_1}(\beta_1)
          \bar{\chi_2}^{\lambda_2}(\beta_2)\\
          =&\frac{1}{e_1e_2}\sum_{{{0\leq \lambda_1\leq e_1-1,}\atop {0\leq \lambda_2\leq e_2-1,}}\atop {\chi_1^{\lambda_1}\chi_2^{\lambda_2}(\f{q}^*)=1}}G_{Q_1}(\chi_1^{\lambda_1})G_{Q_2}(\chi_2^{\lambda_2})\sum_{(\beta_1,\beta_2)\in H,\beta_1\beta_2\neq 0}\bar{\chi_1}^{\lambda_1}(\beta_1)\bar{\chi_2}^{\lambda_2}(\beta_2)
\end{align*}
where $\chi_i$ are the characters of $\f{Q_i}^*$ such that $\chi_i(\gamma_i)=\zeta_{e_i}$, for $i=1,2$. The condition $\chi_1^{\lambda_1}\chi_2^{\lambda_2}(\f{q}^*)=1$ is equivalent to
\[
   \chi_1^{\lambda_1}\chi_2^{\lambda_2}(\delta)=1
\]
where $\delta=\gamma_1^{\frac{Q_1-1}{q-1}}=\gamma_2^{\frac{Q_2-1}{q-1}}$ is a primitive element~\footnote{In general, $\gamma_1^{\frac{Q_1-1}{q-1}}$ and $\gamma_2^{\frac{Q_2-1}{q-1}}$ may be differ from some element of $\f{q}$. But this is not essential. By choosing carefully $\gamma_1$ and $\gamma_2$ at the beginning, we can ensure $\gamma_1^{\frac{Q_1-1}{q-1}}=\gamma_2^{\frac{Q_2-1}{q-1}}$.} of $\f{q}$. So it is equivalent to
\[
  1= \chi_1^{\lambda_1}(\gamma_1^{\frac{Q_1-1}{q-1}})\chi_2^{\lambda_2}(\gamma_2^{\frac{Q_2-1}{q-1}})=\zeta_{e_1}^{\lambda_1\frac{Q_1-1}{q-1}}
  \zeta_{e_2}^{\lambda_2\frac{Q_2-1}{q-1}}
\]
which is also equivalent to
\[
   \lambda_1e_2\frac{Q_1-1}{q-1}+\lambda_2e_1\frac{Q_2-1}{q-1}\equiv 0 \qquad \mod e_1e_2.
\]

\subsection{Application of the Second Formula}
Take $e_1=1,\,e_2=q-1,\,Q_1=q^{k_1},\,Q_{2}=q^{k_2}$ such that $\gcd(k_1,k_2)=1$ and $2\nmid k_2$, then $e'_1=\gcd(e_1,\frac{Q_1-1}{q-1})=1$, $n_1=q^{k_1}-1$,
$n_2=\frac{Q_2-1}{q-1}$, and $\gcd(n_1,n_2)=1$. From $e'_1=1$, we know that
\[
  A=-\frac{n|(\f{Q_1}^*,0)\cap H|}{Q_1-1}=-\frac{(q^{k_2}-1)|(\f{Q_1}^*,0)\cap H|}{q-1}.
\]
Next, we determine the exact value of $B$:
\begin{align*}
    B=&\frac{1}{e_1e_2}\sum_{{{0\leq \lambda_1\leq e_1-1,}\atop {0\leq \lambda_2\leq e_2-1,}}\atop {\chi_1^{\lambda_1}\chi_2^{\lambda_2}(\f{q}^*)=1}}G_{Q_1}(\chi_1^{\lambda_1})G_{Q_2}(\chi_2^{\lambda_2})\sum_{(\beta_1,\beta_2)\in H,\beta_1\beta_2\neq 0}\bar{\chi_1}^{\lambda_1}(\beta_1)\bar{\chi_2}^{\lambda_2}(\beta_2)\\
    =&\frac{-1}{q-1}\sum_{{0\leq \lambda_2\leq q-2,}\atop {\chi_2^{\lambda_2}(\f{q}^*)=1}}G_{Q_2}(\chi_2^{\lambda_2})\sum_{(\beta_1,\beta_2)\in H,\beta_1\beta_2\neq 0}\bar{\chi_2}^{\lambda_2}(\beta_2).
\end{align*}
The condition $\chi_2^{\lambda_2}(\f{q}^*)=1$ has been discussed twice before. One can use the same argument to derive that it is equivalent to
\[
   q-1\mid \lambda_2\frac{Q_2-1}{q-1}
\]
By the assumption $2\nmid k_2$, we have $\gcd(q-1,\frac{Q_2-1}{q-1})=1$. So $\chi_2^{\lambda_2}(\f{q}^*)=1$ is equivalent to $q-1\mid \lambda_2$. But $0\leq \lambda_2\leq q-2$, so the condition $\chi_2^{\lambda_2}(\f{q}^*)=1$ holds if and only if $\lambda_2=0$. Hence, we get
\[
   B=\frac{1}{q-1}\sum_{(\beta_1,\beta_2)\in H,\beta_1\beta_2\neq 0}1=\frac{|(\f{Q_1}^*\times\f{Q_2}^*)\cap H|}{q-1}.
\]
In conclusion, for any $1\leq j\leq k_1$, any $D\in \gbinom{C}{j}$ with $D\cap
C'=\{0\}$, let $H=\varphi(D)\in \gbinom{\f{Q_1}\times\f{Q_2}}{j}$, then $(0,\f{Q_2})\cap H=\{(0,0)\}$ and
\begin{align*}
     N_j(D)=&\frac{n}{q^j}+\frac{1}{q^j}(A+B)\\
           =&\frac{n}{q^j}+\frac{1}{q^j}\left(-\frac{(q^{k_2}-1)|(\f{Q_1}^*,0)\cap H|}{q-1}+\frac{|(\f{Q_1}^*\times\f{Q_2}^*)\cap H|}{q-1}\right)\\
            =&\frac{n}{q^j}+\frac{1}{q^j}\left(-\frac{q^{k_2}|(\f{Q_1}^*,0)\cap H|}{q-1}+\frac{|(\f{Q_1}^*\times\f{Q_2})\cap H|}{q-1}\right)\\
             =&\frac{n}{q^j}+\frac{1}{q^j}\left(-\frac{q^{k_2}|(\f{Q_1},0)\cap H|}{q-1}+\frac{|(\f{Q_1}\times\f{Q_2})\cap H|}{q-1}+\frac{q^{k_2}-1}{q-1}\right)\\
              =&\frac{q^{k_1+k_2}-q^{k_1}+q^j-q^{k_2}|(\f{Q_1},0)\cap H|}{q^j(q-1)}.
\end{align*}
So in order to make $N_j(D)$ large, we need to find $H\in \gbinom{\f{Q_1}\times\f{Q_2}}{j}$ such that $(0,\f{Q_2})\cap H=\{(0,0)\}$, and $|(\f{Q_1},0)\cap H|$ is small.

When $k_1\leq k_2$, take $H$ as the $\f{q}$-spanning space of $\{(v_1,e_1),\cdots,(v_{j},e_j)\}$ where $v_1,\cdots,v_j$ form a partial basis of $\f{Q_1}$ and $e_1,\cdots,e_j$ form a partial basis of $\f{Q_2}$. Then $H\in \gbinom{\f{Q_1}\times\f{Q_2}}{j}$ satisfies $(0,\f{Q_2})\cap H=\{(0,0)\}$ and
\[
|(\f{Q_1},0)\cap H|=1
\]
achieves the minimum. So in this case, we have
\[
  M_j(C,C')=n-\frac{q^{k_1+k_2}-q^{k_1}+q^j-q^{k_2}}{q^j(q-1)}=n+\sum_{\iota=0}^{k_1-j-1}q^\iota-\sum_{\kappa=k_2-j}^{k_1+k_2-j-1}q^\kappa.
\]

When $j\leq k_2<k_1$, take the same space $H$ as above, then we obtain the same result
\[
  M_j(C,C')=n-\frac{q^{k_1+k_2}-q^{k_1}+q^j-q^{k_2}}{q^j(q-1)}=n+\sum_{\iota=0}^{k_2-j-1}q^\iota-\sum_{\kappa=k_1-j}^{k_1+k_2-j-1}q^\kappa.
\]

When $k_2<j\leq k_1$, the canonical isomorphisms of $\f{q}$-linear spaces
\[
  \frac{H}{(\f{Q_1},0)\cap H}\cong \frac{(\f{Q_1},0) + H}{(\f{Q_1},0)}\cong \pi_2(H)
\]
induce $j-\dim_{\f{q}}((\f{Q_1},0)\cap H)=\dim_{\f{q}}(\pi_2(H))\leq k_2$. So
\[
  \dim_{\f{q}}((\f{Q_1},0)\cap H)\geq j-k_2.
\]
On the other hand, we will show the equality is available for some $H$.
Let $m=k_1+k_2-j$, then $k_2\leq m< k_1$. Taking $H$ to be the $\f{q}$-linear space spanning by
$(\alpha_1^i,\alpha_2^i), 0\leq i\leq m-1$, as $m\leq k_1$, we
know $\dim_{\f{q}}(H)=m$ and $(0,\f{Q_2})\cap H=\{(0,0)\}$. In this case, $(\f{Q_1},0)\cap H$ achieves
the minimal dimension $j-k_2$. So
\[
  M_j(C,C')=n-\frac{q^{k_1+k_2}-q^{k_1}+q^j-q^{k_2}q^{j-k_2}}{q^j(q-1)}=n-q^{k_1-j}\sum_{\iota=0}^{k_2-1}q^\iota.
\]

From the discussion above, we have the following corollary as an application of our second formula:
\begin{cor}
 Let $\gamma_1$ and $\gamma_2$ be two primitive elements of $\f{Q_1}$ and
 $\f{Q_2}$, respectively, where $Q_1=q^{k_1}, Q_2=q^{k_2}, \gcd(k_1, k_2)=1$ and $2\nmid k_2$. Take $e_1=1, e_2=q-1$ and $\alpha_i=\gamma_i^{e_i}$ for $i=1,2$. Let $n_i$ be the order of $\alpha_i$, for $i=1,2$. Let $C, C'$ be the
cyclic codes defined by~(\ref{code1}) and~(\ref{code2}),
respectively. Then the length of $C$ and $C'$ is $n=n_1n_2=(Q_1-1)(Q_2-1)/(q-1)$. For $1\leq j\leq k_1$, the $j$th RGHW of $C$ respect to $C'$ is
\[
   M_j(C,C')=n-N_j,
\]
where if $k_1\leq k_2$, then
\[
N_j=\sum_{\kappa=k_2-j}^{k_1+k_2-j-1}q^\kappa-\sum_{\iota=0}^{k_1-j-1}q^\iota;
\]
if $k_2< k_1$, then
\begin{equation*}
   N_j=\left\{%
\begin{array}{ll}
   \sum_{\kappa=k_1-j}^{k_1+k_2-j-1}q^\kappa-\sum_{\iota=0}^{k_2-j-1}q^\iota, & \hbox{if $1\leq j\leq k_2$;} \\
   q^{k_1-j}\sum_{\iota=0}^{k_2-1}q^\iota, & \hbox{if $k_2<j\leq k_1$.} \\
\end{array}%
\right.
\end{equation*}
\end{cor}

Similarly, one can show
\begin{cor}
 Let $\gamma_1$ and $\gamma_2$ be two primitive elements of $\f{Q_1}$ and
 $\f{Q_2}$, respectively, where $Q_1=q^{k_1}, Q_2=q^{k_2}, \gcd(k_1, k_2)=1$ and $2\nmid k_1$. Take $e_1=q-1, e_2=1$ and $\alpha_i=\gamma_i^{e_i}$ for $i=1,2$. Let $n_i$ be the order of $\alpha_i$, for $i=1,2$. Let $C, C'$ be the
cyclic codes defined by~(\ref{code1}) and~(\ref{code2}),
respectively. Then the length of $C$ and $C'$ is $n=n_1n_2=(Q_1-1)(Q_2-1)/(q-1)$. For $1\leq j\leq k_1$, the $j$th RGHW of $C$ respect to $C'$ is
\[
   M_j(C,C')=n-N_j,
\]
where if $k_1\leq k_2$, then
\[
N_j=\sum_{\kappa=k_2-j}^{k_1+k_2-j-1}q^\kappa-\sum_{\iota=0}^{k_1-j-1}q^\iota;
\]
if $k_2< k_1$, then
\begin{equation*}
   N_j=\left\{%
\begin{array}{ll}
   \sum_{\kappa=k_1-j}^{k_1+k_2-j-1}q^\kappa-\sum_{\iota=0}^{k_2-j-1}q^\iota, & \hbox{if $1\leq j\leq k_2$;} \\
   q^{k_1-j}\sum_{\iota=0}^{k_2-1}q^\iota, & \hbox{if $k_2<j\leq k_1$.} \\
\end{array}%
\right.
\end{equation*}
\end{cor}

\section{Conclusions}
In this paper, we give two general formulae for the RGHWs of the cyclic codes with two nonzeros respect to one of its two irreducible cyclic subcodes.
By using the formulae, RGHWs of two explicit classes of cyclic codes are obtained. Comparing with the computation of GHWs of the linear code, RGHWs have
one more restriction which is difficult to describe or satisfy in general when one applies the techniques from computing GHWs. For example, if one wants to
compute the RGHWs of irreducible cyclic codes by using the method in this paper, then the restriction on the intersection with the subcode is hard to describe. So
the method in this paper is not suitable in this scenario. We leave this problem as a future research problem.

\bibliographystyle{IEEEtr}
\bibliography{RGHW}



\end{document}